\documentclass{jgaa-art}
\usepackage{todonotes}

\usepackage{amssymb}
\usepackage{amsmath}

\usepackage{nicefrac}
\usepackage{url}
\usepackage{color}
\usepackage[american]{babel}
\usepackage{graphicx}
\usepackage{caption,subcaption}
\usepackage{cutwin}

\newtheorem{corollary}{Corollary}
\newtheorem{conjecture}{Conjecture}

\newcommand{\leaveout}[1]{}

\renewcommand{\medskip}{\smallskip}
\renewcommand{\int}{{\ensuremath{\rm int\,}}}
\renewcommand{\paragraph}[1]{\smallskip\noindent{\em #1: }}

\date{}
\title{Ideal Tree-Drawings of Approximately Optimal Width (And Small Height)}
\Ack{Research supported by NSERC.} 

\author[second]{Therese Biedl}{biedl@uwaterloo.ca}
\affiliation[second]{David R.~Cheriton School of Computer
Science, University of Waterloo, Waterloo, Ontario N2L 1A2, Canada. }

\begin{document}
\doi{}
\Issue{0}{0}{0}{0}{0}
\HeadingAuthor{T. Biedl}
\HeadingTitle{Ideal Tree-Drawings}

\submitted{June 2016}%
\reviewed{}%
\revised{}%
\accepted{}%
\final{}%
\published{}%
\type{}%
\editor{}%

\maketitle
\begin{abstract}
For rooted trees, an {\em ideal} drawing is one that is planar,
straight-line, strictly-upward, and order-preserving.  
This paper considers ideal drawings of rooted trees with the objective
of keeping the width of such drawings small.  It is not known
whether finding the minimum-possible width is NP-hard or polynomial.
This paper gives a 2-approximation for this problem, and a 
$2\Delta$-approximation
(for $\Delta$-ary trees) where additionally the height is $O(n)$.
For trees with $\Delta\leq 3$, the former algorithm finds ideal drawings with
minimum-possible width.
\end{abstract}

\section{Introduction}

Let $T$ be a rooted tree.
An {\em upward drawing} of $T$ is one in which
the curves from parents to children are $y$-monotone.  
It is called {\em strictly upward} if the
curves are strictly $y$-monotone.  All drawings
must be {\em planar} (no edges cross), and 
{\em order-preserving} (the 
drawing respects a given order of children around a node).
Usually they should be 
{\em straight-line} (edges are drawn as straight-line segments).
A tree-drawing is called an \emph{ideal drawing} \cite{Chan02}
if it is planar, strictly-upward, straight-line, and 
order-preserving.

To keep drawings legible, 
nodes are required to be placed at 
{\em grid-points} (i.e., have integer coordinates),
and the main objective is to minimize the width and height of the
required grid.
In a strictly-upward drawing of a rooted tree, the height can
never be smaller than the (graph-theoretic) height of the tree,
and so may well be required to be $\Omega(n)$.  Hence for such
drawings the main objective is to minimize the width.

\medskip\noindent{\bf Previous Results:}
Any $n$-node tree has a planar straight-line strictly-upward drawing of area
$O(n\log n)$ \cite{CDP92}, but these drawings are not order-preserving.
If we additionally want order-preserving drawings, then the construction
by Chan gives such a drawing of area $O(n4^{\sqrt{2\log n}})$ \cite{Chan02}.
For binary trees, Garg and Rusu 
showed that $O(\log n)$ width and $O(n\log n)$ area can be achieved \cite{GR03}.
This is optimal (within the class of binary trees with $n$ nodes) since there
are binary trees that require width $\Omega(\log n)$ and height $\Omega(n)$
for any upward drawing \cite{CDP92}.
See the recent overview paper by Frati and Di Battista \cite{DF14} for
many other related results.

It is not known whether $O(n\log n)$ area can be achieved for ideal
drawings of rooted trees.  If the condition on straight-line drawings
is relaxed to allow {\em poly-line drawings} (i.e., edges may have bends,
as long as the bends are on grid-points), then a minor modification of
the construction of Chan achieves planar strictly-upward order-preserving
drawings with $O(n\log n)$ area \cite{Chan02}.

It is also not known whether finding minimum-width ideal drawings is NP-hard
or polynomial.  In a recent paper, I showed that finding minimum-width drawings
is feasible if either the ``order-preserving'' or the ``straight-line'' 
condition is dropped \cite{OPTI}, but neither of these two algorithms seems
to generalize to minimum-width ideal drawings.   If ``upward'' is dropped,
then one can minimize the smaller dimension (then usually chosen to be the 
height) for unordered drawings 
\cite{ASR+10} and approximate it for order-preserving drawings
\cite{BB-ArXiV16}.

\medskip\noindent{\bf Results of this paper:}  This paper gives
two approximation-algorithms for the width of ideal tree-drawings.  The
first one is a 2-approximation, which is quite similar to Chan's approach
\cite{Chan02}, but uses the so-called {\em rooted pathwidth $rpw(T)$} (the width
of a minimum-width unordered upward drawing \cite{OPTI}) to find a path
along which to split the tree and recurse.  

However, the method to construct these drawings relies on first constructing
$x$-monotone poly-line drawings and the stretching them into a straight-line
drawing.  This generally results in extremely large
height, and in fact, one can argue that for some trees exponential height
is required for drawings of optimal width.  But for practical purposes, it makes
more sense to be more generous in the width if this reduces the
height drastically.  This motivates the second algorithm of this paper,
which creates drawings whose
width may be a factor $O(\Delta)$ away from the optimum, but
where the height is $O(n)$.  In particular, this gives ideal drawings
of area $O(\Delta n\log n)$; the existence of such drawings was previously
shown only for binary trees by Garg and Rusu \cite{GR03}.    
With a minor modification,
the algorithm achieves width $2rpw(T)-1\leq 2\log(n+1)-1$ for binary trees,
while the one by Garg and Rusu used width up to $3\log n$.

\section{Background}

A {\em rooted tree} $T$ consists of $n$ nodes $V$, of which one has
been selected to be the {\em root}, and all non-root nodes have a
unique {\em parent} in such a way that the root is the ancestor
of all other nodes.  The {\em arity} of  a node is its number of children.
We say that $T$ {\em has arity $\Delta$} if all nodes have arity at most
$\Delta$.  A {\em binary} ({\em ternary}) tree is a tree with arity 2 (3).
A node without children is called a {\em leaf}.
A {\em root-to-leaf path} is a path from the root to some leaf.

For any node $v$, we use $T_v$ to denote the subtree of $T$ consisting of
all descendants of $v$ (including $v$ itself).  
We assume that for each node a specific order of the children has been
fixed.  We usually use $c_1,\dots,c_d$ for the children of the root,
enumerated from left to right.

A {\em drawing} of $T$ maps each node $v$ to a grid-point with integer
coordinates.  The {\em width} ({\em height}) of such a drawing is the
smallest integer $W$ ($H$) such that (after possible translation)
all used grid-points have $x$-coordinate
($y$-coordinate) in $\{1,\dots,W\}$ ($\{1,\dots,H\}$).  The grid-line with
$x$-coordinate ($y$-coordinate) $i$ is called {\em column $i$} (row $i$).
All drawings are required to be {\em planar} (i.e., no two edges cross),
{\em strictly-upward} (i.e., parents have larger $y$-coordinate than their
children) and {\em order-preserving} (i.e., children appear in the prescribed
left-to-right order).  We usually consider {\em straight-line} drawings
where edges are represented by straight-line segments between their endpoints,
but occasionally relax this to {\em poly-line drawings}, where edges may
have bends, as long as these bends are also at grid-points and the curve
of the edge remains strictly $y$-monotone.
We often identify the graph-theoretic concept (node, edge, subtree) with
the geometric feature (point, poly-line, drawing) that represents it.

Crucial for our construction is the so-called {\em rooted pathwidth} $rpw(T)$
of a tree $T$ \cite{OPTI}.  We set $rpw(T):=1$ if $T$ is a path from the
root to a (unique) leaf.  Else, we set $rpw(T):=\min_{P\subset T}
\max_{T'\subset T-P} \left\{ 1+rpw(T') \right\}$, where the minimum is
taken over all root-to-leaf paths in $T$ and the maximum is taken over
all subtrees that remain after removing the nodes of $P$ from $T$.  A 
root-to-leaf path $P$ is called an {\em rpw-main-path} if the above minimum 
is achieved at $P$.  Note the root can have at most one child $c_i$ 
such that $rpw(T_{c_i})=rpw(T)$, because any such child must be in any 
rpw-main-path.
If such a child exists, then we call it the {\em rpw-heavy child} of the
root.  It follows from the lower-bound argument in \cite{CDP92} (and was
shown explicitly in \cite{OPTI}) that any planar upward drawing of a tree $T$
has width at least $rpw(T)$, even if the drawing is neither straight-line
nor order-preserving.

\section{A 2-approximation}
\label{sec:2appox}

This section details an algorithm to create 
straight-line order-preserving drawings of width $2rpw(T)-1$, hence a
2-approximation for the width.  This algorithm is
very similar to the one hinted at
by Chan in his remarks \cite{Chan02}; the only difference is that
we choose the ``heavy'' child to be the rpw-heavy-child, rather
than the one whose subtree is biggest.

In this (and many other later) construction, we first construct
a poly-line drawing with the additional requirements that edges
are drawn $x$-monotonically.  Then we ``straighten out'' such a
drawing to become a straight-line drawing, at the cost of increasing
the width.

\begin{theorem}
\label{thm:approximation}
Any rooted tree $T$ has an order-preserving strictly-upward poly-line drawing of
width at most $2rpw(T)-1$.  Furthermore, every edge is drawn
$x$-monotonically, and the height is at most $2n-\ell(T)$, where
$\ell(T)$ denotes the number of leaves of $T$.  It can be found in
linear time.
\end{theorem}
\begin{proof}
We create two such drawings; one where the root is at the top-left
corner and one where it is at the top-right corner.  Only the first
construction is explained here; the other one is symmetric.
Clearly the claim holds for a single node, so assume that the root
has children. We know that there can be at 
most one child $c_h$ with $rpw(T_{c_h})=rpw(T)$.
Set $W:=2rpw(T)-1$; we aim
to create a drawing within columns $1,\dots,W$.

\paragraph{Case 1: ${c_h}$ is undefined, or ${c_h}=c_1$}
Recursively draw the subtree at each child with the root at the top left 
corner.  Combine these drawings with the ``standard'' construction
of drawing trees already used in \cite{CDP92,Chan02}.  Thus,
place the root in the top left corner.  
Place the drawings of $T_{c_d},\dots,T_{c_2}$, in this order from
top to bottom, flush left in columns $2,\dots,W-1$.  These drawings
fit since $rpw(T_{c_i})\leq rpw(T)-1$ for $i>1$ and hence the drawings
have width at most $W-2$.  Since the root is in column 1 and each $c_i$
(for $i>1$) is in column 2, edges to $c_i$ can be drawn straight-line.
Place the drawing of $T_{c_1}$ below all the other drawings, flush left
with column 1; this fits since it has width at most $2rpw(T_{c_1})-1
\leq 2rpw(T)-1=W$.  We can
connect the edge from the root to $c_1$ going vertically down.
See Fig.~\ref{fig:standard}.  

\begin{figure}[ht]
\hspace*{\fill}
\begin{subfigure}[b]{0.35\linewidth}
\includegraphics[width=\textwidth,page=2]{ordered2Approx.pdf}
\caption{}
\label{fig:standard}
\end{subfigure}
\hspace*{\fill}
\begin{subfigure}[b]{0.35\linewidth}
\includegraphics[width=\textwidth,page=1]{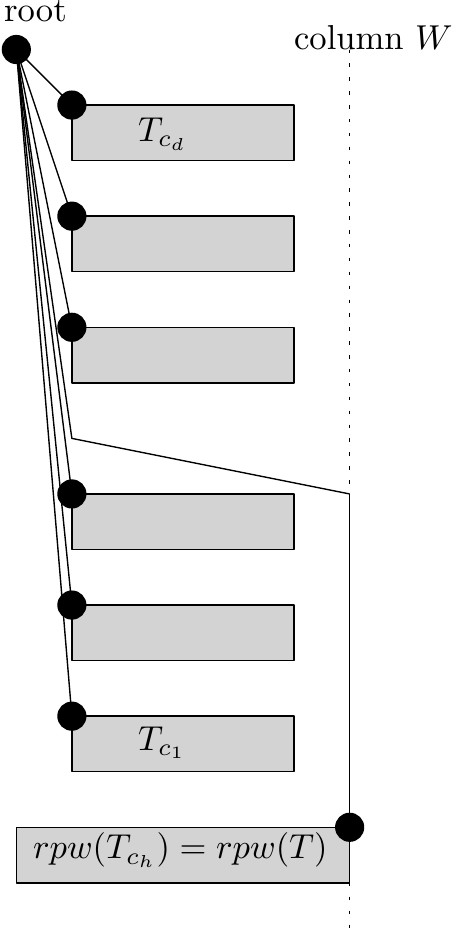}
\caption{}
\label{fig:2approx}
\end{subfigure}
\hspace*{\fill}
\caption{2-approximation algorithm}
\end{figure}

\paragraph{Case 2: $c_h\neq c_1$}
Draw $T_{c_h}$ recursively
with the root in the top right corner, and draw $T_{c_i}$ for $i\neq h$
recursively with the root in the top left corner.
Place the root in the top left corner.  
Place the drawings of $T_{c_d},\dots,T_{c_1}$, in this order from
top to bottom, flush left in columns $2,\dots,W-1$, except omit the drawing
of $T_{c_h}$ and leave one row empty in its place.  As before one argues
 that these drawings fit and that we can connect the root to each $c_i$ for
$i\neq h$.
Place the drawing of $T_{c_h}$ below all the other drawings.  We can
connect the edge from the root to $c_h$ while maintaining the order of
the children by using the empty row between
$T_{c_{h-1}}$ and $T_{c_{h+1}}$, and adding two bends.
See Fig.~\ref{fig:2approx}.

\medskip
In both cases the height of the drawing is the sum of the heights 
of the subtrees, plus one
row for the root and (possibly) one row for the first bend.  Hence it is
at most
$1+\sum_{i=1}^d (2n(T_{c_i})-\ell(T_{c_i})) + 1 =2(n-1) - \ell(T) + 2
=2n-\ell(T)$
as desired.
\end{proof}

\begin{corollary}
\label{cor:general_straight}
Every rooted tree $T$ has an order-preserving strictly-upward straight-line
drawing of width at most $2rpw(T)-1$.
\end{corollary}
\begin{proof}
By the previous theorem $T$ has a strictly-upward order-preserving poly-line drawing
of this width such that edges are drawn $x$-monotonically.  It is known
\cite{EFL96,PT04} that such a drawing can be
turned into a straight-line drawing without increasing the width. 
Neither of these references discusses whether strictly-upward drawings
remain strictly-upward, but it is not hard to show that this can be done:
essentially each subtree needs to ``slide down'' far enough to allow bends
to be straightened out.
\end{proof}

Since $T$ requires width at least $rpw(T)$ in any upward planar drawing
\cite{OPTI}, this gives the desired 2-approximation algorithm.
Since $rpw(T)\leq \log(n+1)$ \cite{OPTI}, this also re-proves
the remark by Chan \cite{Chan02} that trees have order-preserving
upward drawings of area $O(n\log n)$ and straight-line order-preserving
upward drawings of width $O(\log n)$.
Unfortunately the height of these straight-line drawings may be very
large, and so the area is 
no improvement on the area of $O(4^{\sqrt{\log n}} n)$
achieved by Chan \cite{Chan02} for straight-line order-preserving
upward drawings. It remains open to
find such drawings of area $O(n\log n)$
for trees with arbitrary arities.
(For bounded arity, such drawings will be constructed below.)

\subsection{Ternary trees}

For ternary trees, a minor change to the construction yields
optimum width.

\begin{theorem}
\label{thm:ternary}
Every ternary tree $T$ has a poly-line strictly upward order-preserving drawing 
of optimal width $rpw(T)$ and height $\frac{4}{3}n-\frac{1}{3}$ such that 
every edge is drawn $x$-monotonically.
\end{theorem}
\begin{proof}
We show something slightly stronger:  $T$ has such a drawing, and
the root is either placed at the top left or at the top right corner.
The choice between these two corners depends on the structure of the
tree (i.e., it can {\em not} be chosen by the user).
Clearly this holds for a single-node tree $T$, so assume that $T$
consists of a root $v_r$ with children $c_1,\dots,c_d$, in order
from left to right.  Set $W:=rpw(T)$.

Recursively draw each sub-tree $T_{c_i}$ with width $rpw(T_{c_i})$;
note that this draws the sub-tree at the rpw-heavy child with width
at most $W$, and all other sub-trees with width at most $W-1$ by 
definition of rooted pathwidth.  As before we distinguish by the
index of the rpw-heavy child, but in contrast to before we use the
location of the rpw-heavy child in the drawing of the subtree to
determine where to put the root.

\medskip

\paragraph{Case 1: The rpw-heavy child does not exist or is the leftmost child $c_1$}  
In this case the construction is almost exactly as for 
Theorem~\ref{thm:approximation}
(Case 1): the root is in the top-left corner and the subtress are placed
starting in column 2, except for subtree $T_{c_1}$, which occupies all
columns.  However, it may now be that for $i=1,2,3$ tree $T_{c_i}$ has its root $c_i$ in the
top right corner.  If needed, we hence use one bend (and, for $i=2$, an extra row)
to connect from the root to $c_i$; this gives an $x$-monotone drawing.

\paragraph{Case 2: The rpw-heavy child is the rightmost child $c_d$}  
In this case the construction is symmetric:
the root is in the top-right corner.

\paragraph{Case 3: $d=3$ and the rpw-heavy child is child $c_2$}  
We know that in the drawing of $T_{c_2}$ node $c_2$ is placed in one
of the top corners.  

\paragraph{Case 3a: $c_2$ is in the top-right corner}
In this case the construction is similar to the one for Case 2 of
Theorem~\ref{thm:approximation}:
Place the root $v_r$ in the top-left corner, place $T_{c_3}$ in
columns $2,\dots,W$, place bends for edge $(v_r,c_2)$, place $T_{c_1}$
in columns $1,\dots,W-1$, and finally place $T_{c_2}$ and connect the edge
$(u_r,c_2)$.
Note that no additional bend is necessary for $(v_r,c_2)$ since we knew
$c_2$ to be in the top-right corner.

\paragraph{Case 3b: $c_2$ is in the top-left corner}
In this case the construction is symmetric:
the root is in the top-right corner.

\medskip
Clearly the height is at most $\frac{4n-1}{3}=1$ if $n=1$.
If $n>1$ and we needed no extra row for bends, then the height is
at most $1+\sum_{i=1}^d \frac{4}{3}(n(T_{c_i})-1) \leq \frac{4}{3}n-\frac{1}{3}$
since $\sum_{i=1}^d n(T_{c_i})) = n-1$.
If $n>1$ and we did need an extra row for bends, then $d=3$ and therefore
the height is at most 
$$2+\sum_{i=1}^3 \left (\frac{4}{3}(n(T_{c_i})-\frac{1}{3}\right) 
= \frac{4}{3}n-\frac{1}{3}$$
as desired.
\end{proof}

\begin{figure}[ht]
\hspace*{\fill}
\includegraphics[width=22mm,page=1]{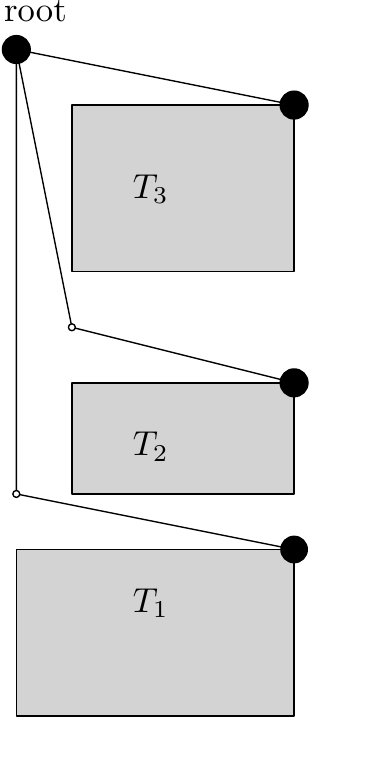}
\hspace*{\fill}
\includegraphics[width=22mm,page=2]{ternary.pdf}
\hspace*{\fill}
\includegraphics[width=22mm,page=3]{ternary.pdf}
\hspace*{\fill}
\includegraphics[width=22mm,page=4]{ternary.pdf}
\hspace*{\fill}
\caption{The four cases for the construction of optimum-width drawings
of ternary trees.}
\label{fig:ternary}
\end{figure}

As before, bends in $x$-monotone curves can be ``straightened out'' by
sliding down, and so we have:

\begin{corollary}
\label{cor:ternary}
Every ternary tree $T$ has a strictly upward order-preserving straight-line drawing of
optimum width $rpw(T)$.
\end{corollary}


\subsection{Bounding the height?}

Notice that Corollary~\ref{cor:ternary} makes no claim on the height.
Indeed, the transformations to straight-line drawings might increase
the height exponentially in general (see \cite{Bie-GD14}), and, as
we show now, also for upward drawings of trees.

\addtocontents{toc}{Optimum-width may have height $(i-1)!$ for ternary\\ }
\begin{theorem}
For any $i\geq 1$,
there exists a ternary tree $T_i$ for which the optimum width of an order-preserving 
upward straight-line drawing is $i$, 
and any such drawing of width $i$ has height at least
$(i-1)!\in n^{\Omega(\log \log n)}$.
\end{theorem}
\begin{proof}
The proof is by induction on $i$.  We define two such trees,
$T_i^R$ and $T_i^L$, that satisfy the following.  Any order-preserving upward drawing of
$T_i^L$ requires width at least $i$, and further, in any drawing of width $i$
the root is in the top left corner and some point in the rightmost column has
vertical distance (i.e., difference in $y$-coordinate) at least $(i-1)!$ from the 
root.  $T_i^R$ is symmetric to $T_i^L$, and hence in any drawing of width $i$ the 
root is at the top right corner and some point in the leftmost column has vertical distance
at least $(i-1)!$ from the root.

For $i=1$, let trees $T_1^R$ and $T_1^L$ consist of a root with one child.  Clearly
this requires width at least 1, and in any drawing of width 1 the root
is in the desired corner.  Since the child cannot be in the same row due
to width 1, and $0!=1$, the child serves as the
node of suitable vertical distance to the root.

For $i\geq 2$, $T_i^L$ consists of the root $u_r$ with three children $c_1,c_2,c_3$.  
Subtree $T_{c_1}$ is a complete binary tree of height $i$ with $2^i-1$ nodes.
Subtrees $T_{c_2}$ and $T_{c_3}$ are the roots of two copies of $T_{i-1}^R$.
See also Figure~\ref{fig:exp_ternary}.

One can easily show that $rpw(T_{c_1})=i$
(since it is a complete binary tree with $2^i-1$ nodes),
so $T_{c_1}$ requires width $i$.   On the other hand, $T_i^L$ has a drawing
of width $i$ (one can stretch the poly-line drawing
in Figure~\ref{fig:exp_ternary}), so its optimal drawing
width is $i$ as desired.

Now fix an arbitrary upward order-preserving straight-line drawing of $T_i^L$ that 
uses exactly $i$ columns.    Since $T_{c_1}$ requires width $i$,
its drawing contains a point $p_1$ in the rightmost column.  The
path from root $u_r$ to $p_1$ must be below the drawings of $T_{c_2}$
and $T_{c_3}$ by the order-property, and hence blocks both $T_{c_2}$ and
$T_{c_3}$ from using the
leftmost column.  

Hence for $k=2,3$, tree $T_{c_k}$ is drawn with width at most $i-1$.  
Since $T_{c_k}=T_{i-1}^R$, therefore induction applies.
So $c_k$ is drawn in the rightmost column (i.e., in column $i$), and the drawing of $T_{c_k}$
contains a point $p_k$ that is in the 
leftmost column of the induced drawing of $T_{c_1}$ (i.e., in column 2) 
and has vertical distance at least $(i-2)!$ from $c_k$.

\begin{figure}[ht]
\hspace*{\fill}
\includegraphics[width=30mm,page=2]{orderedExponential.pdf}
\hspace*{\fill}
\includegraphics[width=30mm,page=1]{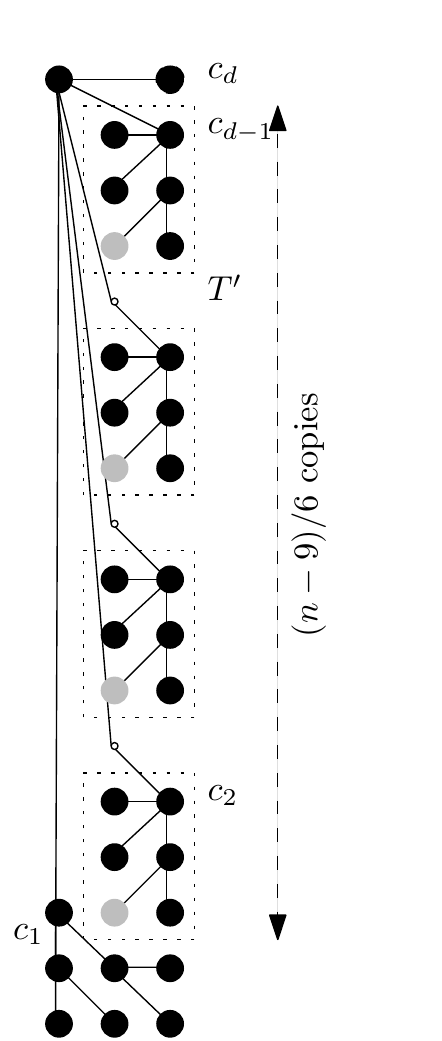}
\hspace*{\fill}
\caption{A ternary tree (left) and a $d$-ary tree (right) that require super-polynomial height in any
optimum-width upward order-preserving straight-line drawing.  For ease of drawing
we add bends to some edges, but the edges are $x$-monotone and hence a
straight-line drawing of the same width exists.}
\label{fig:exp_ternary}
\label{fig:orderedExponential}
\end{figure}

Now we can prove the bound on the height.  
Consider the edge from
the root $u_r$ to $c_2$, which is drawn as a straight-line segment
$\overline{u_r c_2}$.    By order-property and upwardness, $c_3$
must be to the right of $\overline{u_r c_2}$.  By planarity and
upwardness, hence all of $T_{c_3}$ (and in particular node $p_3$)
must be to the right of $\overline{u_r c_2}$.  Since $p_3$ is in
column 2 and $c_2$ is in column $i$, this forces $u_r$ to be in
column 1 as desired.  Furthermore,
$c_2$ must be low enough for $\overline{u_rc_2}$ to be left of $p_3$.
For ease of calculation, translate so that
the root has
$y$-coordinate 0. We hence must have $0\geq y(c_3) \geq y(p_3)+(i-2)!$, and 
hence $\overline{u_rc_2}$ has slope less than $-(i-2)!$.  Since it covers a
horizontal distance of $i-1$, hence the vertical distance of $c_2$
to the root is at least $(i-1)!$ as desired.

This finishes the construction for $T_i^L$, and the one for $T_i^R$
is symmetric with the left and middle child being roots of $T_{i-1}^L$ and
the right child the root of a complete binary tree of height $i$.
It remains to analyze the size of $T_i^L$ and hence obtain the
asymptotic bound.  The number $N(i)$ of nodes of $T_i^L$ and $T_i^R$ 
satisfies the
recursive formula $N(1)=2$ and $N(i)=1+2^{i}-1+2N(i-1)=i\,2^i$. 
Setting $n:=N(i)$, hence $i\geq \log n - \log\log n$ and
for sufficiently large $n$ the required height is at least 
$$ (\log n - \log\log n-1)! 
\geq (\frac{\log n}{4})^{\frac{\log n}{4}} = (2^{\log\log n-2})^{\frac{\log n}{4}}
= n^{\frac{\log\log n-2}{4}}$$
as desired.
\end{proof}

For trees with higher arity, the height-bound can be made asymptotically larger,
essentially by using more copies of tree $T_{i-1}^R$.

\addtocontents{toc}{Optimum-width may have height $2^\Omega(n)$ in general}
\begin{theorem}
There exists $d$-ary $n$-node tree $T$ that has an order-preserving upward straight-line drawing
of width $3$, but any such drawing is required to have height at least
$3\cdot 2^{d-2}= 3\cdot 2^{(n-9)/6}$.
\end{theorem}
\begin{proof}
We construct tree $T$ for $d\geq 4$ and $n=6d-3\geq 21$ 
as follows (see also Fig.~\ref{fig:orderedExponential}):
The leftmost child $c_1$ is the root of a complete binary tree 
with 7 nodes which needs 3 units of width.
The rightmost child $c_{d}$ is a single node.%
\footnote{For strictly-upward drawings, this node can be omitted and the
height lower-bound then becomes $3\cdot 2^{d-1}$ with $d=\frac{n-2}{6}$.}
All other children  $c_2,\dots,c_{d-1}$
are the root of a subtree $T'$ 
that satisfies the following:  $T'$ can be drawn with width
2, but any such drawing requires that the root is in the right
column, and there exists a node $p$ in the left column and at least 
two rows below the root.  One can easily show that
the 6-node tree $T'$ in Figure~\ref{fig:orderedExponential} satisfies this
with the gray node as $p$.

Fix an arbitrary drawing of width 3 of this tree.
Since $T_{c_1}$ requires width 3, there exists a point $p_1$
of $T_{c_1}$ in column 3.
The path from the root to $p_1$ blocks the leftmost column for all
other subtrees, so $T_{c_i}$ for $i>1$ is drawn with width at most 2.
For $1< i< d$, therefore $T_{c_i}$ is drawn with minimum width,
implying that $c_i$ is in the rightmost column and there is a point
$p_i$ in $T_{c_i}$ in column 2 and at least two units below $c_i$.
The goal is to show that the
vertical distance of $p_i$ from the root increases exponentially
with $i$.

After possible translation, assume that the root $u_r$ has
$y$-coordinate $0$. We also know that $u_r$ is in column 1,
because the line-segment $\overline{u_r c_3}$ must bypass
point $p_2$, which is in column 2.
We show that for $1<i<d$ node $p_{d-i}$ must be placed with 
$y$-coordinate at most $-(3\cdot 2^{i}-3)$.
Observe that $c_{d-1}$ is strictly below the root since it is
neither the leftmost nor the rightmost child.  By assumption
$p_{d-1}$ is at least two units below $c_{d-1}$, hence has
$y$-coordinate at most $-3=-(3\cdot 2^{i}-3)$.

For the induction step, assume $p_{d-i+1}$ is placed with $y$-coordinate
$-(3\cdot 2^{i-1}-3)$ for some $i\geq 2$.  
The straight-line segment $\overline{u_rc_{d-i}}$  connects column 1
and 3 and by planarity and order-property
must intersect column 2 at a point {\em below} $p_{d-i+1}$.  Let $Y$ be
the $y$-coordinate of this intersection, then $Y<-(3\cdot 2^{i-1}-3)$
and the $y$-coordinate of $c_{d-i}$ is $2Y<-(3\cdot 2^{i}-6)$.  
Since 
$c_{d-i}$ has integral $y$-coordinate, therefore its $y$-coordinate
is at most $-(3\cdot 2^{i}-5)$.  Since $p_{d-i}$ is two units below,
it has $y$-coordinate
is at most $-(3\cdot 2^{i}-3)$ and the induction holds.

For $p_{2}=p_{d-(d-2)}$, we hence have $y$-coordinate at most 
$-(3\cdot 2^{d-2}-3)$.  Subtree $T_1$ adds at least two more rows
in this column.  
Since the root was at $y$-coordinate 0 and the height counts the number of
rows, the height of the drawing therefore is at least $3\cdot 2^{d-2}$.
The number of nodes in $T$ is $n=1+7+(d-2)6+1=6d-3$, so $d-2=(n-9)/6$
which proves the claim. 
\end{proof}

Our super-polynomial lower bounds on the height requires arity at least 3.
We suspect that such a lower bound also holds for binary
trees, but this remains open.

\begin{conjecture}
There exists a binary tree such that any optimum-width order-preserving upward drawing
has height $\omega(n)$.
\end{conjecture}


\section{A $2\Delta$-approximation with linear height}
\label{sec:height}

In 2003, Garg and Rusu \cite{GR03} showed that every binary tree
has an upward straight-line drawing of width $O(\log n)$ and height at 
most $n$.  However, their construction does not generalize to higher arity
(unless one drops ``upward'').    We now give a different construction
that achieves these bounds for any tree that has constant arity.

\begin{theorem}
\label{thm:overhang}
Every rooted tree $T$ 
has a strictly-upward order-preserving straight-line drawing of width
$(2\Delta-1) (rpw(T)-1)+1$ and height at most $n$, where $\Delta$
is the maximum number of children of a node.  It can be found in
linear time.
\end{theorem}

In particular any rooted tree has a strictly-upward order-preserving
straight-line drawing of area $O(\Delta n\log n)$; this
is an improvement over the
area-bound of $O(4^{\sqrt{\log n}} n)$ by Chan \cite{Chan02} 
for small (but more than constant) values of $\Delta$.


%
%
%
\begin{proof}
For ease of description, define shortcuts $r:=rpw(T)$
and $W(i) := (2\Delta-1) (i-1) + 1$; we aim to create
drawings of width at most $W(r)$.
As before we create drawings where the root is in the top-left corner,
and a symmetric construction places the root in the top-right corner.


If $r=1$ then $W(1)=1$ and $T$ is a path from the root to a single leaf.
We can draw $T$ in a single column as desired.
So assume $r>1$, which means that $\Delta\geq 2$ and 
that the root has children
$c_1,\dots,c_d$, $1\leq d\leq \Delta$.
Let $c_h$ be the child, if any, with
$rpw(T_{c_h})=r$.

\paragraph{Case 1: $c_h$ does not exist, or $c_h=c_1$}
In this case, draw the tree as in the ``standard'' construction,
i.e., recursively obtain drawings of each $T_{c_j}$, $j=1,\dots,d$,
with $c_j$ in the top-left corner
and combine as in Fig.~\ref{fig:standard}).
The drawing of $T_{c_1}$ has width at most $W(r)$
and the drawing of each $T_{c_j}$ for $j>1$ has width
at most $W(r{-}1)\leq W(r)-1$, to which we add at most one unit
width.  Clearly all conditions are satisfied.



\paragraph{Case 2: $c_h\neq c_1$}
The construction in this case is much more complicated (and quite
different from Garg and Rusu's).
%
We use $W(r) = W(r{-}1)+2\Delta-1$ columns for our drawing,
and split them into 3 groups as follows:
\begin{itemize}
\item The leftmost $\Delta-1$ columns are called {\em left-detour} columns.  
	The rightmost of the left-detour columns is called the {\em left-overhang} column.  
\item The next $W(r{-}1)+1$ columns are the {\em middle} columns; the
leftmost and rightmost of the middle columns are called the {\em left-path} 
and {\em right-path} column, respectively.  
\item The last $\Delta-1$ columns are called the {\em right-detour} column.
	The leftmost of the right-detour columns is called the {\em right-overhang} column.
\end{itemize}
Fig.~\ref{fig:orderedOverhang} sketches the construction.
The main tool is to use a rpw-heavy path $P=v_0,v_1,v_2,\dots$
Note that $v_1$ must be child $c_h$,
and so in particular $v_1$ is not the leftmost child of $v_0$ by
case assumption.

We first outline the idea.  
To place path $P$, we split it into many sub-paths of length at least
2.  These sub-paths are alternatingly placed in the left-path column (or
nearby) and the right-path column (or nearby).  Whenever possible, subtrees
of these paths are placed in the middle columns.  However, this is not
always possible for the top-most and bottom-most node of a sub-path.  For
these, we use the detour-columns, either for placing the node or for placing
its children.  However, the subtrees at these nodes or children cannot
be placed here; instead we put them
``much farther down'',
namely, at such a time when path $P$ has veered to the other side
and therefore the middle columns are accessible.

\begin{figure}[ht]
\centering
\includegraphics[height=160mm,page=1]{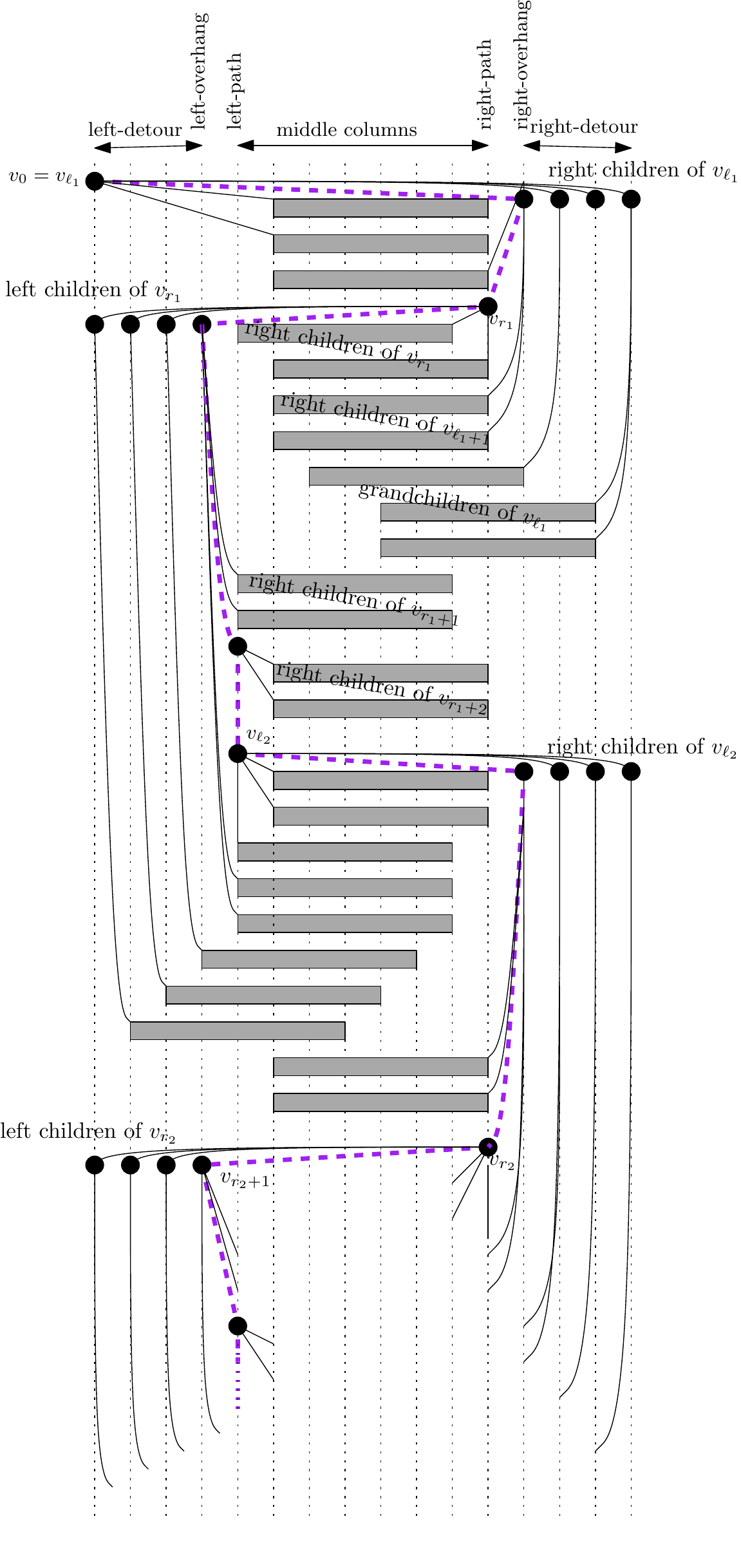}
\includegraphics[height=160mm,page=2,trim=100 0 80 0,clip]{orderedOverhangNoAlpha.pdf}
\captionof{figure}{The construction for order-preserving straight-line drawings
if the rpw-heavy child of the root is not the leftmost child.  Path $P$ is
	purple and dashed.  (Left) The construction for arbitrary $\Delta$.
(Right) The modified version for $\Delta=2$. }
\label{fig:orderedOverhang}
\end{figure}

The precise placement of path $v_0,v_1,v_2,\dots$ is as follows.
Place the root $v_0$ in the top left corner, set $i=1$ and $\ell_1=0$.
(Generally $\ell_i$ will be the index of the bottommost node of the
$i$th sub-path on the left, and $r_i$ will be the index of the bottommost
node of the $i$th sub-path on the right.)  Now repeat:
\begin{itemize}
\item 
$v_{\ell_i+1}$ is placed in the right-overhang column, one row below
	$v_{\ell_i}$.
\item 
$v_{\ell_i+2}$ is placed in the right-path column, some rows below.%
\footnote{``Some rows below'' means ``so that this node is below all
the subtrees that need to be inserted above it by later steps''.  For
this particular situation here, this is the height of
the drawings of the subtrees at left children of $v_{\ell_i}$ and $v_{\ell_{i+1}}$
and (for $i>1$) at children of left children of $v_{r_{i-1}}$.  }

\item 
While $v_j$ is the rightmost
child of $v_{j-1}$ (for $j=\ell_i+3,\ell_i+4,\dots$), 
place it in the right-path column, some
rows below.  
\item 
Let $r_i\geq \ell_i+2$ be the maximal index for which $v_{r_i}$
	was placed in the right-path column.
	So $v_{r_i+1}$ is not the rightmost child of $v_{r_i}$.

\item 
Place $v_{r_i+1}$ in the
left-overhang column, one row below $v_{r_i}$.

\item 
Place $v_{r_i+2}$ in the left-path column, some rows below.
\item 
While $v_j$ is the leftmost
child of $v_{j-1}$ (for $j=r_i+3,r_i+4,\dots$), 
place it in the left-path column, some
rows below.  

\item 
Let $\ell_{i+1}\geq r_i+2$ be the maximal index for which $v_{\ell_{i+1}}$
	was placed in the left-path column.

\item 
Update $i:=i+1$, and repeat until we reach the end of path $P$.
\end{itemize}

For any non-leaf node $v$ on $P$, let the {\em left} [{\em right}] 
children of $v$ be all those children of $v$ that are strictly 
left [right] of the 
child of $v$ on $P$.
We now explain how to place all the subtrees at
right children of nodes $v_{\ell_i},\dots,v_{\ell_{i+1}-1}$, for
$i=1,2,\dots$.
The subtrees at left children are placed symmetrically.
\begin{enumerate}
\item We start at $v_{\ell_i}$.
The right children of $v_{\ell_i}$ are placed,
in order, in the row below $v_{\ell_i}$ and in distinct right-detour columns.  
By choice of $\ell_i$ (or, for $i=1$, by case assumption)
node $v_{\ell_i+1}$ is {\em not} the leftmost child of $v_{\ell_i}$.
So $v_{\ell_i}$ has at least one left child, therefore
at most $\Delta-2$ right children, which means that there are sufficiently
many right-detour columns for placing the right children as well as 
$v_{\ell_i+1}$.
Since these children are one row below $v_{\ell_i}$,
we can connect them to $v_{\ell_i}$ with a straight-line segment
(drawn curved in Fig.~\ref{fig:orderedOverhang}
for increased visibility.)  The subtrees at these children
are {\em not} being placed yet; this will happen in Step~\ref{st:defer}.
\item
The next node is $v_{\ell_i+1}$, which is in the right-overhang column
one row below $v_{\ell_i}$.
The subtrees at its right children will be placed in Step~\ref{st:defer}.
\item
The next nodes are $v_{\ell_i+2},\dots,v_{r_i-1}$.  By choice of
$r_i$ these nodes do not have right children.  The rows for these
nodes (as well as $v_{r_i}$) 
are determined by the symmetric version of Step~\ref{st:no_left} that places subtrees
at left children.  
\item
The next node is $v_{r_i}$, placed in the right-path column.
We place the subtrees at its right children with the symmetric
version of the standard construction of Fig.~\ref{fig:standard}.
Thus, recursively obtain for each such subtree a drawing of width at most
$W(i-1)$ with the child in the top-right corner.
Place these, in order, in the rows below 
$v_{r_i}$ and in the columns to its left (except for the last child,
which shares the column with $v_{r_i}$). 
This fits
within the middle columns since there are $W(r{-}1)+1$ middle columns
and $v_{r_i}$ is in the rightmost of these.
\item Next comes node $v_{r_i+1}$, in the row below $v_{r_i}$ and
	the left-overhang column.  This node might share a row
	with some right child of $v_{r_i}$ but uses a different column.
	The subtrees at $v_{r_i+1}$'s right children will 
be will be placed later (in Step~\ref{st:defer}).
\item	\label{st:defer}
Now we place all the subtrees that were deferred earlier.  First, 
draw the subtrees at right children of $v_{\ell_i+1}$ recursively with their
roots in the top-right corner.  Place these drawings,
flush right with the right-path column, below all the trees of
right children of $v_{\ell_i}$.
Recall that 
$v_{\ell_i+1}$ was placed in the right-overhang column while its children
are now in the right-path column, which is adjacent. Hence 
the edges can be drawn with straight-line segments (shown again with curves 
in Fig.~\ref{fig:orderedOverhang}).

Next, we place the subtrees at right children of $v_{\ell_i}$,
parsing them in left-to-right order. 
If $c$ is such a child, then $c$ was placed much higher up already
in one of the right-detour columns.
Let $g_1,\dots,g_d$ be the children of $c$ 
(hence grand-children of $v_{\ell_i}$).  For each $g_i$ create
a drawing of $T_{g_i}$ with $g_i$ in the top-right corner.
Place these drawings in the middle columns as well as the right-detour
columns so that $g_1,\dots,g_d$ are one column to the left of $c$.
Then $c$ can be connected with straight lines (shown again with curves).

Finally place the subtrees at right children of $v_{r_i+1}$.
Recursively draw each such subtree 
with the root in the top-left corner.
Place these, in order, flush left with the left-path column,
and draw the edges to $v_{r_i+1}$ as straight-line segments.
\item
\label{st:no_left}
Next come nodes $v_j$ for $j=r_i+2,r_i+3,\dots,\ell_{i+1}-1$.
For each $j$, place $v_j$ in the next row (i.e., the first row 
below what was drawn so far) and in the left-path column.
Recursively draw the subtree at
each right child of $v_{j}$ with the root in the top-left corner.
Place these, in order, flush left with the column that is one right
of the left-path column.
\item
Finally put $v_{\ell_{i+1}}$ in the next row; and go to Step 1~with
	the next $i$.
\end{enumerate}

This ends the description of the construction, which has width
$W(r)$ as desired.  All rows contain nodes, so the height is at
most $n$.   
%
\end{proof}

\subsection{The special case of binary trees}

We note that for binary trees, our construction gives a width of at most
$3rpw(T)$, hence a 3-approximation.  This can be turned into a
2-approximation by decreasing the number of middle columns. 

\begin{corollary}
Every rooted binary tree $T$ 
has a strictly-upward order-preserving straight-line drawing of width
$2rpw(T)-1$ and height at most $n$.  
\end{corollary}
\begin{proof}
Set $W'(r)$ to be the recursive function $W'(1)=1$ and
$W'(r)=W'(r-1)+2$ (which resolves to $W'(r)=2r-1$).
Apply exactly the same 
construction as before, using $\Delta-1=1$ overhang columns on
each side, but use only $W'(r-1)$ middle columns.
See also Fig.~\ref{fig:orderedOverhang}(right). 

It remains to argue that we can do the construction using one less
middle column per recursion.  We show here only that the subtrees
at right
children ``fit''; the argument is symmetric for the left children.
This can be seen (for $i=1,2,\dots$) as follows:
\begin{itemize}
\item Node $v_{\ell_i}$ has no right child since the tree is binary
	and by choice of $\ell_i$ it has a left child.
\item The subtree at the right child of $v_{\ell_i+1}$ has width
	at most $W'(r-1)$ by induction.  This fits into the middle
	columns.  We can connect the child to $v_{\ell_i+1}$ since
	the latter is in the right-overhang column, i.e., in the
	adjacent column.
\item Node $v_j$ with $\ell_i+2\leq j \leq r_i-1$ has no right child
	by choice of $r_i$.
\item Node $v_{r_i}$ has at most one right child since the tree is binary.  
	This child is placed vertically below $r_i$, and hence its subtree
	can use all middle columns.
\item Consider node $v_j$ for $j$ with $r_i+2\leq j \leq \ell_{i+1}-1$, which
	is placed in the left-path column.  The subtree at its right child 
	may use $W'(r-1)$ columns, but only $W'(r-1)-1$ of the middle
	columns are available to it, since the leftmost column is used
	by $v_j$ and edge $(v_j,v_{j+1})$.  However, at this $y$-range no
	node or edge uses the right-overhang column, so we can use the
	right-overhang column to place the subtree at the right child.
\end{itemize}
Hence the construction works, for binary trees, with only $W'(rpw(T))$
middle columns, and the width is hence at most $W'(rpw(T))\leq 2rpw(T)-1$.
\end{proof}

\section{Conclusion}
\label{sec:conclusion}

This paper gave approximation algorithms for the width of 
strictly-upward order-preserving drawings of trees.  It was shown
that one can approximate the width within a factor of 2 (and even
find the optimum width for ternary trees), albeit at the cost of
a very large height.  A second construction gave drawings with linear
height for which the width is within $2\Delta$ of the optimum.
In particular this implies ideal drawings of area $O(n\log n)$
for all trees with constant arity.

Among the most interesting open problems is whether it is possible
to find the minimum width of ideal tree-drawings in polynomial time.
Secondly, what can be said if the height should be small?
Does every rooted tree have a strictly-upward
	straight-line order-preserving drawing of area $O(n\log n)$, or is
it possible to prove a lower bound of $\omega(\log n)$ if (say) at most $n$
rows may be used?

\bibliographystyle{plain}
\bibliography{journal,full,gd,papers}

\end{document}